\documentclass[11pt]{amsart}
\usepackage{latexsym,amssymb,amsmath}
\newtheoremstyle{custom}
  {3pt}
  {3pt}
  {\slshape}
  {}
  {\bfseries}
  {.}
  { }
   {}
\theoremstyle{custom}
\newtheorem{theorem}{Theorem}[subsection]

\newtheorem{proposition}[theorem]{Proposition}
\newtheorem{proposition/definition}[theorem]{Proposition/Definition}

\newtheorem{corollary}[theorem]{Corollary}

\theoremstyle{definition}

\theoremstyle{remark}
\newtheorem{remark}[theorem]{Remark}

\newtheoremstyle{exercise}
  {3pt}
  {6pt}
  {}
  {}
  {\bfseries}
  {:}
  { }
   {}
\theoremstyle{exercise}
\newtheorem{exercise}[theorem]{Exercise}
\newtheoremstyle{exercises}
  {3pt}
  {6pt}
  {}
  {}
  {\bfseries}
  {:}
  {\newline}
   {}

\theoremstyle{exercise}
\newtheorem{exercises}[theorem]{Exercises}

\input epsf
\def\boxit#1{\vbox{\hrule height1pt\hbox{\vrule width1pt\kern3pt
  \vbox{\kern3pt#1\kern3pt}\kern3pt\vrule width1pt}\hrule height1pt}}

\def\BC{\mathbb C}

\def\BP{\mathbb P}
\def\pp#1{\mathbb P^{#1}}

\def\fgl{\mathfrak g\mathfrak l}

\def\pp#1{{\mathbb P}^{#1}}
\def\tdim{{\rm dim}}

\def\hd{,...,}
\def\ww{\wedge}

\def\inv{{}^{-1}}

\def\11{\mathbf 1}

\def\fsl{{\mathfrak {sl}}}

\def\l{\lambda}
\def\a{\alpha}

\def\b{\beta}

\def\s{\sigma}

\def\ot{{\mathord{ \otimes } }}

\def\ra{{\mathord{\;\rightarrow\;}}}

\def\La#1{\Lambda^{#1}}

\def\frak{\mathfrak}
\def\fgl{\frak g\frak l}\def\fsl{\frak s\frak l}

\def\s{\sigma}

\def\a{\alpha}
\def\b{\beta}

\def\l{\lambda}

\def\BP{\mathbb  P}
\def\BC{\mathbb  C}

\def\pp#1{\mathbb  P^{#1}}

\def\hd{, \hdots ,}

\def\inv{{}^{-1}}

\def\La#1{\Lambda^{#1}}

\def\pp#1{\mathbb  P^{#1}}

\def\ur{\underline {\bold R}}

\def\ra{\rightarrow}

\def\tdet{\operatorname{det}}

\def\tdim{\operatorname{dim}}

\def\tmod{\operatorname{mod}}

\def\tmin{\operatorname{min}}

\def\ww{\wedge}

\def\bbb{{\bold{b}}}

\def\be{\begin{equation}}
\def\ene{\end{equation}}
\def\aaa{{\bold {a}}}
\def\bbb{{\bold {b}}}
\def\ccc{{\bold {c}}}

\def\tlog{{\rm{log}}}

\numberwithin{equation}{section}
\numberwithin{theorem}{section}

\def\ll#1#2{\lambda_{#1,#2}}
 
\begin{document}

\title[Nontriviality of equations for border rank]{Nontriviality of equations and explicit tensors in $\BC^m\ot \BC^m\ot \BC^m$ of border rank at least $2m-1$}
\author{J.M. Landsberg}
 \begin{abstract} For odd $m$, I write down   tensors in $\BC^m\ot \BC^m\ot \BC^m$ of border rank at least  $2m-1$, showing the
non-triviality of the Young-flattening equations of \cite{LOsecbnd} that vanish on the matrix multiplication tensor. I also study the border rank of the tensors
of \cite{alexeev+forbes+tsimerman:2011:tensor-rank} and \cite{MR88g:15021}. I   show  the tensors $T_{2^k}\in \BC^k\ot \BC^{2^k}\ot \BC^{2^k}$,
of \cite{alexeev+forbes+tsimerman:2011:tensor-rank}, despite having
rank equal to $2^{k+1}-1$, have border rank equal to $2^k$. I show the equations for border rank of \cite{MR88g:15021} on
$\BC^m\ot \BC^m\ot \BC^m$
  are trivial in the case of border rank $2m-1$ and determine their precise non-vanishing on the matrix multiplication tensor.
\end{abstract}
\thanks{ Landsberg  supported by NSF grant  DMS-1006353}
\email{jml@math.tamu.edu}
\maketitle

\section{Results and context}
Let $A,B,C$ be complex vector spaces of dimensions $\aaa,\bbb,\ccc$. A tensor $T\in A\ot B\ot C$ is said to
have {\it rank one} if $T=a\ot b\ot c$ for some $a\in A,b\in B,c\in C$. More generally the {\it rank} of a tensor $T\in A\ot B\ot C$ is
the smallest $r$ such that $T$ may be written as the sum of $r$ rank one tensors. Let $\hat \s_r^0\subset A\ot B\ot C$ denote
the set of tensors of rank at most $r$. This set is not closed (under taking limits or in the Zariski topology) so let
$\hat \s_r$ denote its closure (the closure is the same in the Euclidean or Zariski topology).
The  variety   $\hat \s_r$ is
familiar   in algebraic geometry,   it is   cone over the $r$-th secant variety of
the Segre variety, but we won't need that in what follows.
The rank and border rank of a tensor are measures of its complexity.  While rank is natural to complexity theory,
border rank is more natural from the perspective of geometry, as one can obtain lower bounds on border rank via
polynomials. Let $\bold R(T)$,  $\ur(T)$ respectively denote the rank and border rank of $T$.

The maximum rank of a tensor in $\BC^m\ot \BC^m\ot \BC^m$ is
at most $m^2$, although it is not known in general if this actually occurs.
The maximum border rank of a tensor in $\BC^m\ot \BC^m\ot \BC^m$ is $\lceil \frac{m^3}{3m-2}\rceil$ for all
$m\neq 3$ and five when $m=3$, see \cite{Strassen505,MR87f:15017}. 
It is an important problem to find explicit tensors of high rank and border rank, and to develop tests that
bound the rank and border rank from below. For the border rank, such tests are in the form of polynomials that vanish
on $\hat \s_r$. For rank the study is more complicated. All lower bounds for rank that I am aware of arise from
first proving a lower bound on border rank, and then taking advantage of special structure of a particular
tensor to show its rank is higher than its border rank.

Perhaps the most important tensor for this study
is the matrix multiplication tensor, where one considers matrix multiplication
$$
M_n: \BC^{n^2}\times\BC^{n^2}\ra \BC^{n^2} 
$$
  as a tensor $M_n\in \BC^{n^2*}\ot \BC^{n^2*}\ot \BC^{n^2}=A\ot B\ot C$.
In \cite{LOsecbnd}, G. Ottaviani and I proved the bound $\ur(M_n)\geq 2n^2-n$ by finding polynomials that vanished on  
$\hat\s_{2n^2-n-1}$ and showing these polynomials did not vanish on $M_n$. At the same
time we found additional polynomials that vanished on  
$\hat\s_{2n^2-k}$ for $k=n\hd 2$, but these polynomials also vanished on $M_n$. However {\it  we did not know whether or not these additional
polynomials were identically zero}. The motivation for this paper was to show these additional
 polynomials are in fact not identically zero. To do this I write down an explicit sequence  tensors
on which the polynomials do not vanish, see Theorem \ref{genthm}.   These are the first proven nontrivial polynomials   for
border rank  in $\BC^m\ot \BC^m\ot \BC^m$ beyond $2m-\sqrt{m}$. Since     matrix multiplication satisfies these polynomials, 
 the result raises the intriguing possibility that the border rank of matrix multiplication could be far less than I had previously expected.

My first hope had been to use the tensors  of \cite{alexeev+forbes+tsimerman:2011:tensor-rank}, as they had been shown to have high rank, but it turns out,
see Proposition  \ref{tnppr} below, that they have low border rank.  

A referee for an earlier version of this paper  wrote that   
\cite{MR88g:15021}  
contains equations for $\hat \s_r$ in the range
$m+1\leq r\leq 2m-1$. This turned out
to   be  erroneous  - the author of \cite{MR88g:15021} had only claimed
the equations were {\it potentially}  nontrivial in this range. Since the equations are presented indirectly, it was difficult to
determine their non-triviality in general (see \S\ref{griessect} for a discussion), but I do show:
\begin{proposition}\label{griesprop} 
Let $\tdim A=\aaa$, $\tdim B=\tdim C=m$. Then Griesser's equations of \cite{MR88g:15021} for $\hat \s_r$
have the following properties:
\begin{enumerate}
\item They are trivial for $r=2m-1$ and all $\aaa$.
\item They are trivial for $r=2m-2$, $\aaa=m$ and $m\leq 4$.
\item Setting $m=n^2$, matrix multiplication $M_n$ fails to satisfy the equations for $r\leq \frac 32 n^2-1$ when $n$ is even
and $r\leq \frac 32 n^2+\frac n2-2$ when $n$ is odd, and satisfies the equations for all larger $r$.
\end{enumerate}
\end{proposition}

I was unable to determine whether or not
  the equations  are trivial for $r=2m-2$,  $\aaa=m$ and $m>4$. If they are nontrivial for even $m$, they 
  would give equations beyond the equations of \cite{LOsecbnd}.

  In \cite{MR88g:15021} the equations  are only shown to be nontrivial on matrix multiplication for $r\leq n\lceil\frac {3n}2\rceil   - 2$  
and their non-triviality in general was not examined. Note that the bound for $n$ odd that (3) gives is   $\ur(M_n)\geq \frac 32 n^2+\frac 12 n-1$,
which equals
  Lickteig's bound  of \cite{MR86c:68040} which held the \lq\lq world record\rq\rq\ for over twenty years. 
\smallskip

  The equations of \cite{LOsecbnd} are special cases of equations obtained
  via {\it Young flattenings} defined in \cite{LOver}, which I now review.
The classical flattenings  (which date back at least to Macaulay and Sylvester) arise by viewing $T\in A\ot B\ot C$ as a linear map
$T: B^*\ra A\ot C$, and taking the size $(r+1)$ minors (i.e., the determinants of the $(r+1)\times (r+1)$-submatrices),
which give   equations for $r\leq \tmin(\bbb,\aaa\ccc)$. These do not give all the equations and the
idea behind Young flattenings is to pass from multi-linear algebra to linear algebra in
more sophisticated ways. The particular Young flattening used in \cite{LOsecbnd} may be described
as follows:

Let $\La pA \subset A^{\ot p}$  denote the skew-symmetric   tensors.
Let $Id_{\La p A}: \La p A\ra \La pA$ denote the identity map, and consider, assuming $p\leq \lfloor\frac \aaa 2\rfloor$, the
map $T\ot Id_{\La p A}: B^*\ot \La p A\ra \La{p}A\ot A\ot C$. Compose this map with the skew-symmetrization map to 
get a map
\be\label{ourmap}
T_{A}^{\ww p}:  \La p A\ot B^*\ra \La{p+1}A\ot C.
\ene
If $\bold R(T)=1$, then the linear map $T_{A}^{\ww p}$ has rank $\binom{\aaa-1}p$.
More precisely, if $T=a\ot b\ot c$, then the image of $T_A^{\ww p}$ is 
the image of $\La p A\ot a\ot c$ under the skew-symmetrization map
$\La p A\ot A\ot C\ra \La{p+1}A\ot C$. 
Thus if   $\ur(T)\leq r$, then the size 
$\binom{\aaa -1}p r+1$ minors of $T_{A}^{\ww p}$ will be zero. These minors
 are the equations used in \cite{LOsecbnd} to bound the border rank of matrix multiplication.

Now let $\aaa=\bbb=\ccc=m$, so when dealing with matrix multiplication, $m=n^2$.  The Young flattenings  (potentially) give
the best lower bounds when $p=\lfloor\frac m2\rfloor$ so I examine them in that range. If $m=2p+1$,
$\frac {m^2}{m-p}=2m-2+\frac 1{p+1}$ and if $m=2p+2$, then $\frac {m^2}{m-p}=2m-4+\frac 4{p+2}$, so
the minors of \eqref{ourmap}  potentially give  equations for $\hat \s_r$  up to $r=2m-2$ when $m$ is odd and $r=2m-4$ when $m$ is even.
In \cite{LOsecbnd} it was shown these equations are nontrivial (i.e., do not vanish identically) when $m$ is a square up  to
$r= 2m- \sqrt{m} $ by showing they did not vanish on the matrix multiplication tensor
$M_n\in \BC^{n^2}\ot \BC^{n^2}\ot\BC^{n^2}$.

\begin{theorem}\label{genthm} When $m$ is odd and equal to $2p+1$,  the 
maximal minors of \eqref{ourmap}  give  nontrivial equations for $\hat \s_r\subset \BC^m\ot\BC^m\ot \BC^m$, the
tensors of border rank at most $r$ in $\BC^m\ot \BC^m\ot \BC^m$, up to $r=2m-2$. They give equations up to  $r=2m-4$ when $m$ is even.
The maximal minors  do not vanish on the explicit tensors $T_m(\l)\in \BC^m\ot \BC^m\ot \BC^m$ of  \eqref{tnlambda}.\end{theorem}

Thus the equations may be used to show that $\ur(T)\geq 2m-1$ when $m$ is odd and $\ur(T)\geq 2m-3$ when $m$ is even.
These are the largest values of border rank  we know  how to test for.

\begin{corollary} Let $M_n\in \BC^{n^2}\ot \BC^{n^2}\ot \BC^{n^2}$ denote the matrix multiplication operator.
Then $M_n$ satisfies    nontrivial equations for the variety of tensors of border rank $2n^2-n+1$.
\end{corollary}

\smallskip

In \cite{alexeev+forbes+tsimerman:2011:tensor-rank} (also see \cite{weitzpreprint}), setting $m=2^k$,  they give an explicit sequence of tensors $T_m\in \BC^{k+1}\ot \BC^m\ot \BC^m$ of
rank $2m-1$,
see \eqref{alextena}, and explicit tensors $T_{m+1}'\in \BC^{m+1}\ot \BC^{m+1}\ot \BC^m$ of rank $3(m+1)-k-4$, see  \S\ref{tnpsect}. Their tensors
may be defined over an arbitrary field.  

\begin{proposition}\label{tnpr} Let $m=2^k$.  The tensors $T_m\in \BC^{k+1}\ot \BC^m\ot \BC^m$ of \eqref{alextena} have border rank $m$,
i.e., $\ur(T_m)=m<\bold R(T_m)=2m-1$. 
\end{proposition}

\begin{proposition}\label{tnppr} Let $m=2^k$.  The tensors $T_{m+1}'\in \BC^{m}\ot \BC^{m+1}\ot \BC^{m+1}$ of \S\ref{tnpsect}  
 satisfy 
  $m+2\leq \ur(T_{m+1}')\leq 2(m+1)-2-k <\bold R(T_{m+1}')=3(m+1)-4-k$. 
\end{proposition}
 
 I expect the actual border rank to be close to the lower bound as many of the Young flattening
 equations vanish, even in the $p=1$ case.

\subsection*{Acknowledgments} I thank L. Manivel and G. Ottaviani for useful conversations.
I also thank C. Ikenmeyer and anonymous referees of an earlier version of this paper for many useful suggestions,
in particular a referee pointed  out a much simpler proof of the border rank of the tensors in \cite{alexeev+forbes+tsimerman:2011:tensor-rank}
than the one I had given.

\section{Proof of Theorem \ref{genthm}} \label{pfasect}
Let $a_{-p}\hd a_{p}$ be a basis of $A=\BC^{2p+1}$, $b_1\hd b_m$ a basis of $B=\BC^m$,  and $c_1\hd c_m$ a basis of $C=\BC^m$.

Let $\ll iu$ be numbers satisfying open conditions to be specified below.  (They may be chosen to be e.g., $\ll iu=2^{2^{i+u}}+2^u$.)   Consider
\be \label{tnlambda}
T_{m,p}(\l):=\sum_{j=-p}^{-1}
a_j\ot (\sum_{\a=1}^{m-p+j-1}\l_{-j,\a}b_{j+p+1+\a}\ot c_{\a}) + \sum_{j=0}^p a_j\ot (\sum_{\b=1}^{m-j} b_{\b}\ot c_{j+\b})
\ene
When $m=2p+1$, write $T_m(\l)=T_{2p+1,p}(\l)$. 

For example, in matrices, when $p=1$ and $m=3$
$$
T_3(\l)= a_{-1}\ot \begin{pmatrix} 0 &0 &0\\ \l_{1,1}  &0&0\\ 0&\l_{1,2}  &0 
\end{pmatrix}+a_0\ot \begin{pmatrix} 1&0&0\\ 0 &1&0\\ 0&0 &1 
\end{pmatrix}+a_1\ot \begin{pmatrix} 0 & 1&0\\   0&0 & 1\\ 0& 0 & 0 
\end{pmatrix}.
$$

Write $T=\sum a_j\ot X_j$, where $X_j\in B\ot C$,  and use $X_0=\sum_{\a=1}^m b_\a\ot c_\a: B^*\ra C$ to identify $C$ with $B^*$, 
so  $X_0$ becomes the identity matrix in $\fgl(B)$, the space of endomorphisms of $B$.
Let    $([X_i,X_j])$, $i,j\in \{ -p\hd -1,1\hd p\}$ denote the $2mp\times 2mp$ block matrix, whose $(i,j)$-th block is the commutator $[X_i,X_j]=X_iX_j-X_jX_i$.
By \cite[p. 4]{Lrank},
\eqref{ourmap} is injective if and only if  
\be\label{coorform}
\tdet_{2mp}([X_i,X_j])\neq 0.
\ene
\begin{remark} Despite the simplicity of the equation \eqref{coorform}, I do not know how to prove it is related to border rank other than
by  expressing \eqref{ourmap} in coordinates, making a choice of   $a_0$, and applying elementary identities
regarding determinants. It would be desirable to have a direct explanation.
\end{remark}

The choice of bases and $X_0$  gives a   grading  to
   $\fgl(B)$. That is, one has a vector space decomposition
$\fgl(B)=\oplus_{j=-m}^m \fgl(B)_j$,  
 such that the commutators satisfy $[\fgl(B)_i,\fgl(B)_j]\subset \fgl(B)_{i+j}$. Here $\fgl(B)_j$ consists of the matrices that
 are zero except on the $j$-th diagonal (with $j=0$ being the main diagonal).
With this grading, taking $T=T_{m,p}(\l)$, 
  $X_j\in \fgl(B)_j$.  
In particular, omitting the zero index,  $[X_i,X_j]$,  is an $m\times m$  matrix that is zero if $i,j$ are greater than zero,
and otherwise zero except for the $(i+j)-th$ diagonal, all of whose entries are nonzero as long as for each fixed $i$, the $\ll iu$ are distinct
as $u$ varies.
Writing the $2mp\times 2mp$ matrix,  ordered $p,p-1\hd 1,-1\hd  -p$, as four equal size square blocks, the first block,   consisting
of the upper left $mp\times mp$ submatrix, is zero, so
the  determinant of $([X_i,X_j])$  is, up to sign,  the square of the determinant of the lower left $mp\times mp$ submatrix.

I will show that this determinant may be written as a product of smaller determinants
  of sizes $1,2\hd p-1,p,p\hd p,p-1,p-2\hd 2,1$.

For example, when $p=2$ and $m=5$ we get (blocking $2,1,-1,-2$ for both rows and columns) the lower left block consists of four smaller blocks which are:

$$[X_2,X_{-1}]=
\begin{pmatrix}
0&\l_{1,2}&0&0&0   \\
0&0&\l_{1,3}-\l_{1,1}&0&0 \\
0&0&0&\l_{1,4}-\l_{1,2}&0\\
0&0&0&0&-\l_{1,3}     \\
0&0&0&0&   0       
\end{pmatrix} 
$$

$$
[X_2,X_{-2}]=\begin{pmatrix}
\l_{2,1}         &0&0&0&0   \\
0&\l_{2,2}         &0&0&0  \\
0&0&\l_{2,3}-\l_{2,1}&0&0   \\
0&0&0&        -\l_{2,2}&0   \\
0&0&0&0&   -\l_{2,3}       
\end{pmatrix}
$$

$$
[X_{1},X_{-1}]=\begin{pmatrix}
  \l_{1,1}&0&0&0&0           \\
  0&\l_{1,2}-\l_{1,1}&0&0&0  \\
  0&0&\l_{1,3}-\l_{1,2}&0&0  \\
  0&0&0&\l_{1,4}-\l_{1,3}&0  \\
  0&0&0&0&   -\l_{1,4}       
\end{pmatrix} 
$$

$$
[X_1,X_{-2}]=
\begin{pmatrix}
   0&0&0&0&0                  \\
  \l_{2,1} &0&0&0&0           \\
  0& \l_{2,2}-\l_{2,1}&0&0&0  \\
  0&0& \l_{2,3}-\l_{2,2}&0&0  \\
  0&0&0& -\l_{2,31}&0         
\end{pmatrix}.
$$
In this case the determinant of the lower left block is 
\begin{align*}&(\l_{2,1})
\tdet\begin{pmatrix} \l_{1,2}  &\l_{1,1} \\ \l_{2,2 } & \l_{2,1} \end{pmatrix}
\tdet\begin{pmatrix} \l_{1,3}-\l_{1,1} &\l_{1,2}-\l_{1,1}\\ \l_{2,3}-\l_{2,1}& \l_{2,2}-\l_{2,1}\end{pmatrix}\\
&\ \ \cdot
\tdet\begin{pmatrix} \l_{1,4}-\l_{1,2} &-\l_{2,2} \\ \l_{1,3}-\l_{1,2}& \l_{2,3}-\l_{2,2}\end{pmatrix} 
\tdet\begin{pmatrix} -\l_{1,4}  &-\l_{2,3} \\ \l_{1,4}-\l_{1,3}&  -\l_{2,1}\end{pmatrix}(-\l_{1,4}).
\end{align*}

Returning to the general case, consider the first column of the $mp\times mp$ matrix $([X_i,X_{-k}])$, $1\leq i,k\leq p$.  All the entries are zero except
the first entry of the lowest  block, i.e., the entry in the slot $((p-1)m+1,1)$, which is $\ll p1$.

Now consider the second column. There are two nonzero entries - the first entry of the second lowest block, which is $\l_{p-1,2}$, and the
second entry of the last block, which is $\l_{p,2}$. The $(m+1)$-st column (first column of the second block)
 also has two nonzero entries, and they occur at the same heights,
the entries are respectively $\l_{p-1,1}$ and $\l_{p,1}$.
Thus these two columns contribute $\tdet\begin{pmatrix} \l_{p-1,2} &\l_{p-1,1} \\ \l_{p,2}& \l_{p,1}\end{pmatrix}$ to the
determinant.

Consider the third column. If $p>2$, there are three nonzero entries, the first 
entry of the third to last block, the second entry of the second to last block,
and the third entry  of the last block. The second column of the second block and the third column of the third block all have the
same nonzero entries. The result is a contribution of
$$
\tdet\begin{pmatrix} \l_{p-2,3} &\l_{p-2,2} & \l_{p-2,1} \\  \l_{p-1,3}& \l_{p-1,2}& \l_{p-1,1}\\  \l_{p ,3}& \l_{p ,2}& \l_{p ,1}
\end{pmatrix}
$$
to the determinant.

In general, considering the $i$-th column, for $i\leq p$,   there is a contribution of the determinant of
an $i\times i$  matrix whose $(s,t)$-th entry is $\l_{p-i+1+s,i+t}$,  
or  $\l_{p-i+1+s,i+t}- \l_{p-i+1+s,i+t+1}$, or $- \l_{p-i+1+s,i+t+1}$.
For $i\geq p$ one gets $p\times p$ matrices until they start shrinking in size.

In all cases, for  any given minor of size $f$ that appears, it will have a unique term  with coefficient plus or minus one on
$\Pi_{s=1}^f\l_{p-i+1-s,i+s}$, so for a generic choice of $\l$ it will not vanish.
Thus for a generic choice no minors will vanish, which means that their product, the determinant,
will not vanish either, proving the theorem.  
To have an explicit matrix, one could take e.g., $\l_{i,j}=2^{2^{i+j}}+2^j$ to assure a single monomial in each minor
will dominate the expression.

\section{The tensors $T_m$ of \cite{alexeev+forbes+tsimerman:2011:tensor-rank}}

I restrict to the case $m=2^k$ because the other cases are similar only padded with zeros.
In \cite{alexeev+forbes+tsimerman:2011:tensor-rank} they define tensors $T_m\in \BC^{k+1}\ot \BC^m\ot \BC^m=A\ot B\ot C$ by
\be\label{alextena}
T_m:=a_0\ot(\sum_{\b=1}^m b_{\b}\ot c_{\b})+ \sum_{j=1}^k a_j\ot ( \sum_{\a=1}^{2^{j-1}} b_{\a}\ot c_{m-2^{\a-1}+1})
\ene
Here  I have changed the indices slightly from \cite{alexeev+forbes+tsimerman:2011:tensor-rank}.

For example, when $k=3$, in matrices, this is:
$$
T_8(A^*)=\begin{pmatrix}
a_0 & & & & & & & \\
 & a_0& & & & & & \\
 & &a_0 & & & & & \\
 & & &a_0 & & & & \\
a_3 & & & &a_0 & & & \\
 & a_3& & & &a_0 & & \\
a_2 & &a_3 & & & &a_0 & \\
a_1 &a_2 & &a_3 & & & &a_0 
\end{pmatrix}.
$$

If we  reorder the basis of $B$  and write the tensor as
\be\label{tnsym}
T_m=a_0\ot(\sum_{\b=1}^m \tilde b_{m-\b}\ot c_{\b})+ \sum_{j=1}^k a_j\ot ( \sum_{\a=1}^{2^{j-1}} \tilde b_{m-\a}\ot c_{m-2^{\a-1}+1})
\ene
We see this is a specialization of
the multiplication tensor in $\BC[X]/(X^m)$ whose border rank is   $m$ (see, e.g.,  \cite[Ex. 15.20]{BCS}).
This proves Proposition \ref{tnpr}.

\section{The tensors $T_{m+1}'$ of \cite{alexeev+forbes+tsimerman:2011:tensor-rank}}\label{tnpsect}
In \cite{alexeev+forbes+tsimerman:2011:tensor-rank}, they also  define tensors in $\BC^{m+1}\ot \BC^{m }\ot \BC^{m+1}$ by enlarging 
the matrices to have size $m\times (m+1)$ and adding
vectors   in the last column.
For example, when $k=3$ (so $m=8$), one gets the $8\times 9$ matrix
$$
T_9'(A^*):=
\begin{pmatrix}
a_0 & & & & & & & &a_4\\
 & a_0& & & & & & &a_5\\
 & &a_0 & & & & & &a_6\\
 & & &a_0 & & & & &a_7\\
a_3 & & & &a_0 & & & &a_8\\
 & a_3& & & &a_0 & & \\
a_2 & &a_3 & & & &a_0 & \\
a_1 &a_2 & &a_3 & & & &a_0 \\
\end{pmatrix}
$$
which they express as a tensor in $\BC^{m+1}\ot \BC^{m+1}\ot \BC^{m+1}$ by adding zeros.
These tensors have rank close to $3m$, to be precise
$\bold R(T_m')=3m-2H(m-1)-\lfloor \tlog_2(m-1)\rfloor -2$, where $H(m)$ is the number of $1$'s in the binary expansion of
$m$, so the rank is best if $m-1=2^k$, in which case
$\bold R(T_{2^k+1}')=3(2^k+1)-4-k$.
The border rank is   smaller.  Write $T_{m+1}'=T_m+T_m''$ where $T_m''=(a_{k+1}\ot b_1+a_{k+2}\ot b_2+\cdots +a_{m}\ot b_{m-k})\ot c_{m+1}$.
  Thus $\ur(T_{m+1}')\leq \ur(T_m)+\ur(T_m'')= m+m-k$.
One obtains the lower bound of $m+2$ for the border rank (as opposed to the trivial $m+1$) because
the map $T_A^{\ww 1}$ has a kernel of size $2^{k-1}=\frac {m-1}2$. Since this kernel is still quite large,
I expect the actual border rank to be close to the lower bound.

 \section{The equations of \cite{MR88g:15021}}\label{griessect}
 Given $T=\sum_{j=0}^{\aaa-1} a_j\ot X_j$  with $a_j$ a basis of $A$, $X_j\in B\ot C$,  $\tdim A=\aaa$ and $\tdim B=\tdim C=m$,  
 assume  $X_0$ is of full rank and use it  to identify  $C$ with $B^*$.
 The equations in \cite{MR88g:15021} are stated as:  if the border rank of $T$ is at most $r$, with
 $m+1\leq r\leq 2m-1$, then the space of endomorphisms
 $\langle  [X_1,X_2]\hd [X_1,X_{\aaa -1}] \rangle \subset \fsl(B)$ is such that there exists $E\in G(2m-r,B)$,
 with $\tdim(\langle  [X_1,X_2]\hd [X_1,X_{\aaa -1}] \rangle(E))\leq r-m$. (Here $\langle ....\rangle$ denotes
 the linear span and $G(k,B)$ the Grassmannian of $k$ planes   in $B$.) Compared with the equations \eqref{coorform}, here one is just examining the last block column
 of the matrix appearing in \eqref{coorform}, but one is extracting apparently more refined information from it.

 Assuming $T$ is sufficiently generic, we may choose $X_1$ 
  to be diagonal with distinct entries on the diagonal (a general element of $\fsl(B)$, the space of
  traceless endomorphisms, is diagonalizable
 with distinct eigenvalues), and  this is a generic choice of $X_1$.
 Let   $\fsl(B)_R$ denote  the
matrices with zero on the diagonal (the sum of the root spaces). 
 Then 
 $ad(X_1): \fsl(B)_R\ra \fsl(B)_R$, given by $Y\mapsto [X,Y]$,  is a linear  isomorphism, and $ad(X_1)$ kills the diagonal matrices.
 Write $U_j=[X_1,X_j]$, so the $U_j$ will be matrices with zero on the diagonal, and
by picking $T$ generically we can have any such matrices, and this is the most
general choice of $T$ possible, so if the equations vanish for a generic choice
of $U_j$, they vanish identically.
 
 \begin{proof}[Proof of Proposition \ref{griesprop}]
 
 Proof of (1): 
 In the case $r=2m-1$, so $r-m=m-1$ and $\aaa\leq m+1$ the equations are trivial as we only have $\aaa-2\leq m-1$ linear maps. 
 When $\aaa\geq m+2$ a na\"\i ve dimension count makes it possible for the equations to be non-trivial, the
 equations are that $\tdim \langle U_2v\hd U_{\aaa-1}v\rangle\leq m-1$. However, with our
 normalizations of $X_0=Id$ and $X_1$  diagonal  with distinct  entries on the diagonal, taking $v=(1,0\hd 0)^T$ (the superscript  $T$   denotes transpose), the $U_jv$ will be contained in the hyperplane
 of vectors with their first entry zero. Since we only made genericity assumptions, we conclude.
  
  \smallskip
  
  Proof of (2): In the case $r=2m-2$, the equations will be nontrivial if and only if there exist
  $U_2\hd U_{\aaa-1}$ such that   for all  linearly independent $v,w$  
  $\tdim \langle U_2v\hd U_{\aaa-1}v, U_2w\hd U_{\aaa-1}w \rangle\geq  m-1$.
  For $\aaa=m$, we saw we could have $U_2v\hd U_{m-1}v$ linearly independent, so
  the nontriviality  condition is that for some $j$, $U_jw\not\in \langle U_2v\hd U_{m-1}v\rangle$.
  
  First observe that
  $U_jw \in \langle U_2v\hd U_{m-1}v\rangle \tmod U_j\hat v$ (where $\hat v$ is the line determined by $v$)  means
  $w=\sum_{k\neq j} a_{j,k}U_j\inv U_k v$ for some constants $a_{j,k}$. (We are working with generic $U_j$ so we may
  assume they are invertible.)
  Thus we must have $v$ and constants  $s_{i,j},t_{i,j}$, such that
  $s_{i,j}\sum_{k\neq j} a_{j,k}U_j\inv U_k v= t_{i,j}\sum_{l\neq i} a_{l,i}U_l\inv U_i v$, i.e.,
  $U_2\hd U_{m-1}$ must be such that there exist constants $s_{i,j},t_{i,j}$ for  $i<j$, and  $a_{j,k}$ for   $j\neq k$ such that 
  $$\tdet(s_{i,j}\sum_{k\neq j} a_{j,k}U_j\inv U_k - t_{i,j}\sum_{l\neq i} a_{l,i}U_l\inv U_i)=0.
  $$
  
  When $m=4$, the $s,t$ are irrelevant and we need $a_{2,3}U_2\inv U_3v=a_{3,2}U_3\inv U_2v$, i.e., 
  that for some choice of $[a_{2,3},a_{3,2}]\in\pp 1$, the linear map
  $a_{2,3}U_2\inv U_3-a_{3,2}U_3\inv U_2$ has a kernel. But every $\pp 1$
of matrices intersects the hypersurface $\tdet_m=0$ so we conclude.

\begin{remark} I expect the equations are non-trivial for $m\geq 5$ but I was unable to
show this, even for $m=5$.
The $r=2m-1$ case shows that one should be cautious.
Consider the $m=5$ case. 
The equations would be trivial if for all $U_2,U_3,U_4\in \fsl(B)_R$,   one could choose
$([a_{2,3},a_{2,4}], [a_{2,3},a_{2,4}],[s_{2,3},t_{2,3}],[s_{2,4},t_{2,4}])\in \pp 1\times \pp 1\times \pp 1\times \pp 1$ 
such that the linear maps
$ s_{2,3}(a_{2,3}U_2\inv U_3+a_{2,4}U_2\inv U_4)-t_{2,3}(a_{3,2}U_3\inv U_2+a_{3,4}U_3\inv U_4) $
and $ s_{2,4}(a_{2,3}U_2\inv U_3+a_{2,4}U_2\inv U_4)-t_{2,4}((a_{4,2}U_4\inv U_2+a_{4,3}U_4\inv U_3)) $
have a common kernel. If we consider the variety $\Sigma_m\subset G(m-3,\BC^{m^2})$ defined by 
$$
\Sigma_m:=\{ E\in G(m-3,\BC^{m^2})\mid \exists v\in V\backslash 0 \ {\rm such\ that}\ e.v=0 \forall e\in E\},
$$
then $\tdim \Sigma_m=(m-1)+(m-3)(m^2-m-(m-3))$  (as for each point in $\BP V$ there is an $m^2-m$ dimensional
space of endomorphisms with the line in the kernel, and we have the Grasssmannian of $m-3$ planes in that
space of endomorphisms). So in the $m=5$ case   a general four dimensional subvariety
of the Grassmanian will fail to intersect   $\Sigma_5$, but our four dimensional subvariety is not general.
\end{remark}

Proof of (3): 
Consider  matrix multiplication $M_n\in \BC^{n^2}\ot \BC^{n^2}\ot \BC^{n^2}=A\ot B\ot C$. 
With a judicious choice of bases, $M_n(A)$ is block diagonal
\be\label{blockform}
\begin{pmatrix} x & & \\ & \ddots & \\ & & x\end{pmatrix}
\ene
where $x=(x^i_j)$ is $n\times n$. In particular, the image is closed under brackets. 
Choose $X_0$ so it is the identity. We may not have $X_1$ diagonal with distinct entries on the diagonal, the best we can
do is for $X_1$ to be block diagonal with each block having the same $n$ distinct entries.    For a  subspace $E$
of dimension $2m-r=dn+e$ (recall $m=n^2$) with $0\leq e\leq n-1$, the image of a generic choice of  $[X_1,X_2]\hd [X_1,X_{n^2-1}]$  applied to $E$ is of dimension
at least $(d+1)n$ if $e\geq 2$, at least $(d+1)n-1$ if $e=1$ and $dn$ if $e=0$, 
and  equality will hold if we choose $E$ to be, e.g., the span of the  first $2m-r$ basis vectors of $B$.
(This is because the $[X_1,X_j]$ will span the entries of type \eqref{blockform} with zeros on the diagonal.)
If $n$ is even, taking $2m-r=\frac{n^2}2+1$, so $r=\frac {3n^2}2-1$, the image occupies a space of dimension
$\frac {n^2}2+n-1>\frac{n^2}2-1=r-m$. If one takes $2m-r=\frac{n^2}2$, so  $r=\frac {3n^2}2$, the image occupies a space of dimension
$\frac {n^2}2=r-m$, showing Griesser's equations cannot do better for $n$ even.
If $n$ is odd, taking $2m-r=\frac {n^2}2-\frac{ n}2+2$, so $r=\frac{3n^2}2+\frac{n}2-2$, the image will have dimension
$\frac{n^2}2+\frac n2>r-m=\frac{n^2}2+\frac n2-1$, and taking $2m-r=\frac{n^2}2-\frac n2+1$
the image can have dimension $\frac{n^2}2-\frac n2+(n-1)=r-m$, so the equations vanish for this and all
larger $r$. Thus Griesser's equations for $n$ odd give Lickteig's bound
$\ur(M_n)\geq \frac{3n^2}2+\frac{n}2 -1$.
\end{proof}
   
\bibliographystyle{amsplain}
 
\bibliography{Lmatrix}

\end{document}